\makeatletter
  \providecommand*\input@path{}
  \newcommand\addinputpath[1]{
  \expandafter\def\expandafter\input@path
  \expandafter{\input@path{#1}}}
  \addinputpath{importfile/}
  \newcommand{\Rmnum}[1]{\expandafter\@slowromancap\romannumeral #1@}
  \newcommand{\thickhline}{%
    \noalign {\ifnum 0=`}\fi \hrule height .6pt
    \futurelet \reserved@a \@xhline
  }
  \def\email#1{\emailname:#1}
  \def\emailname{E-mail}%

  \def\acknowledgement{\par\addvspace{17pt}\small\rmfamily
  \trivlist\if!\ackname!\item[]\else
  \item[\hskip\labelsep
  {\bfseries\ackname}]\fi}
  
  \newenvironment{acknowledgements}{\begin{acknowledgement}}
  {\end{acknowledgement}}
  \def\trans@english{\switcht@albion}
  \def\switcht@albion{\def\ackname{Acknowledgements}%
  }\switcht@albion
\makeatother

\documentclass[12pt]{article}

\usepackage[left=2.50cm,right=2.50cm,top=2.50cm,bottom=2.50cm]{geometry} 
\usepackage{mathrsfs}
\usepackage{float}
\usepackage{cite}
\usepackage[misc]{ifsym}
\usepackage{makecell}
\usepackage{amsmath} 
\usepackage{amsfonts} 
\usepackage{amssymb} 
\usepackage{amsthm}
\usepackage[english]{babel}
\usepackage{graphicx}   
\usepackage{pict2e}
\usepackage{url}       
\usepackage{bm}         
\usepackage[justification=centering]{caption}
\usepackage{multirow}
\usepackage{booktabs}
\usepackage{tabu}
\usepackage{algorithm}
\usepackage{algorithmicx}
\usepackage{algpseudocode}
\usepackage[colorlinks,linkcolor=blue,anchorcolor=black,citecolor=blue]{hyperref}

\hypersetup{colorlinks,bookmarks,unicode}
\hypersetup{urlcolor=blue}
\def\keywordname{{\bfseries Keywords}}%
\def\keywords#1{\par\addvspace\medskipamount{\rightskip=0pt plus1cm
\def\and{\ifhmode\unskip\nobreak\fi\ $\cdot$
}\keywordname\enspace\ignorespaces#1\par}}

\newcommand{\tl}{\textnormal}
\newcommand{\rank}{\textrm{Rank}}
\newcommand{\supp}{\textnormal{Supp}}

\newtheorem{theorem}{Theorem}
\newtheorem{corollary}{Corollary}
\newtheorem{lemma}[theorem]{Lemma}
\newtheorem{proposition}{Proposition}

\theoremstyle{remark}
\newtheorem{definition}[theorem]{Definition}
\newtheorem{remark}{Remark}

\date{}

\begin{document}

\title{Polynomial-Time Key Recovery Attack on the Lau-Tan Cryptosystem Based on Gabidulin Codes}

\author{
Wenshuo Guo\thanks{Wenshuo Guo is with the Chern Institute of Mathematics and LPMC, Nankai University, Tianjin 300071, China. \email{ws\_guo@mail.nankai.edu.cn}} 
\ and 
Fang-Wei Fu\thanks{Fang-Wei Fu is with the Chern Institute of Mathematics and LPMC, Nankai University, Tianjin 300071, China. \email{fwfu@nankai.edu.cn}}
}

\maketitle

\begin{abstract}
This paper presents a key recovery attack on the cryptosystem proposed by Lau and Tan in a talk at ACISP 2018. The Lau-Tan cryptosystem uses Gabidulin codes as the underlying decodable code. To hide the algebraic structure of Gabidulin codes, the authors chose a matrix of column rank $n$ to mix with a generator matrix of the secret Gabidulin code. The other part of the public key, however, reveals crucial information about the private key. Our analysis shows that the problem of recovering the private key can be reduced to solving a multivariate linear system over the base field, rather than solving a multivariate quadratic system as claimed by the authors. Solving the linear system for any nonzero solution permits us to recover the private key. Apparently, this attack costs polynomial time, and therefore completely breaks the cryptosystem.
\end{abstract}

\section{Introduction}
\noindent In post-quantum era, most public key cryptosystems based on number theoretic problems will suffer serious security threat. To resist quantum computer attacks, people have paid much attention to seeking alternatives in the future. Among these alternatives, code-based cryptography is one of the most promissing candidates. The security of these cryptosystems rely on the difficulty of decoding general linear codes. The first code-based cryptosystem was the one proposed by McEliece in 1978, which is now called the McEliece cryptosystem \cite{mceliece1978public}. Although this scheme still remains secure, it had never been used in practical situations due to the drawback of large key size. To tackle this problem, various improvements for McEliece's original scheme were proposed one after another. Generally these improvements can be divided into two categories: one is to substitute Goppa codes used in the McEliece system with other families of codes endowed with special structures, the other is to use codes endowed with the rank metric. 

In 1991, Gabidulin et al. proposed an encryption scheme based on rank metric codes, which is now known as the GPT cryptosystem\cite{gabidulin1991ideals}. An important advantage of rank-based cryptosystems lies in their compact representation of public keys. Some representative variants based on the rank metric Gabidulin codes can be found in \cite{gabidulin2003reducible,berger2004designing,faure2005new,loidreau2017new,lau2019new}. Unfortunately, most of these variants, including the original GPT cryptosystem, have been completely broken due to the inherent structural weakness of Gabidulin codes. Specifically, Gabidulin codes contain a large subspace invariant under the Frobenius transformation, which provides the feasibility for one to distinguish Gabidulin codes from general ones. Based on this observation, various structural attacks \cite{overbeck2008structural,horlemann2018extension,Otmani2018Improved,gaborit2018polynomial,coggia2020security} on the GPT cryptosystem and some of their variants were designed. 

In \cite{lau2018new}, Lau and Tan proposed a public key encryption scheme based on Gabidulin codes. In their cryptosystem, the published information consists of a generator matrix of the disturbed Gabidulin code by a random code that admits maximum rank weight $n$ and a random vector of column rank $n$. This technique of masking the structure of Gabidulin codes, as claimed by Lau and Tan, can prevent some existing attacks \cite{overbeck2008structural,horlemann2018extension,gaborit2018polynomial,Otmani2018Improved}. Additionally, the recent Coggia-Couvreur attack \cite{coggia2020security} and Ghatak's attack \cite{2020Extending} do not work on this cryptosystem either.

\textbf{Our contributions}. This paper mainly investigates the security of the Lau-Tan cryptosystem and presents a simple yet efficient key recovery attack on this encryption scheme. Additionally, our analysis shows that all the generating vectors of a Gabidulin code, together with the zero vector, form a $1$-dimensional linear space. In other words, for a fixed generating vector $\bm{g}$ of a Gabidulin code $\mathcal{G}\subseteq\mathbb{F}_{q^m}^n$, any other generating vector of $\mathcal{G}$ must be of the form $\gamma\bm{g}$ for some $\gamma\in\mathbb{F}_{q^m}^*$. This suggests that there are totally $q^m-1$ generating vectors for a Gabidulin code over $\mathbb{F}_{q^m}$. Meanwhile, we also introduce a different approach from the one proposed in \cite{horlemann2018extension} to compute the generating vector of Gabidulin codes when an arbitrary generator matrix is given.

The rest of this paper is organized as follows. Section \ref{section2} introduces some basic notations used throughout this paper, as well as the concept of Moore matrices and Gabidulin codes. Section \ref{section3} gives a simple description of the Lau-Tan cryptosystem. In Section \ref{section4}, we mainly describe the principle of our attack. To do this, we first introduce some further results about Gabidulin codes that will be helpful for explaining why this attack works. Following this, we present this attack in two steps. Additionally we also give a complexity analysis of this attack and some experimental results using Magma. In Section \ref{section5}, we will make a few concluding remarks.

\section{Preliminaries}\label{section2}
\noindent In this section, we first introduce some notations in finite field and coding theory used throughout this paper. After that, we will recall some basic concepts about Gabidulin codes and some related results.

\subsection{Notations and basic concepts}
\noindent For a prime power $q$, we denote by $\mathbb{F}_q$ the finite field with $q$ elements, and $\mathbb{F}_{q^m}$ an extension field of $\mathbb{F}_q$ of degree $m$. Note that $\mathbb{F}_{q^m}$ can be seen as a linear space over $\mathbb{F}_q$ of dimension $m$. A vector $\bm{a}\in\mathbb{F}_{q^m}^m$ is called a basis vector if components of $\bm{a}$ form a basis of $\mathbb{F}_{q^m}$ over $\mathbb{F}_q$. In particular, we call $\bm{a}$ a normal basis vector if $\bm{a}$ is of the form $(\alpha^{q^{m-1}},\alpha^{q^{m-2}},\ldots,\alpha)$ for some $\alpha\in\mathbb{F}_{q^m}^*=\mathbb{F}_{q^m}\backslash\{0\}$. For two positive integers $k$ and $n$, denote by $\mathcal{M}_{k,n}(\mathbb{F}_q)$ the space of all $k\times n$ matrices over $\mathbb{F}_q$, and by $\tl{GL}_n(\mathbb{F}_q)$ the set of all invertible matrices in $\mathcal{M}_{n,n}(\mathbb{F}_q)$. For a matrix $M\in\mathcal{M}_{k,n}(\mathbb{F}_q)$, denote by $\langle M\rangle_q$ the linear space spanned by rows of $M$ over $\mathbb{F}_q$.

An $[n,k]$ linear code $\mathcal{C}$ over $\mathbb{F}_{q^m}$ is a $k$-dimensional subspace of $\mathbb{F}_{q^m}^n$. The dual code of $\mathcal{C}$, denoted by $\mathcal{C}^\perp$, is the orthogonal space of $\mathcal{C}$ under the usual inner product over $\mathbb{F}_{q^m}^n$. A $k\times n$ matrix $G$ is called a generator matrix of $\mathcal{C}$ if its row vectors form a basis of $\mathcal{C}$ over $\mathbb{F}_{q^m}$. A generator matrix $H$ of $\mathcal{C}^\perp$ is called a parity-check matrix of $\mathcal{C}$. For a codeword $\bm{c}\in\mathcal{C}$, the rank support of $\bm{c}$, denoted by $\supp(\bm{c})$, is the linear space spanned by components of $\bm{c}$ over $\mathbb{F}_q$. The rank weight of $\bm{c}$ with respect to $\mathbb{F}_q$, denoted by $\tl{rk}(\bm{c})$, is defined to be the dimension of $\supp(\bm{c})$ over $\mathbb{F}_q$. The minimum rank distance of $\mathcal{C}$, denoted by $\tl{rk}(\mathcal{C})$, is defined to be the minimum rank weight of all nonzero codewords in $\mathcal{C}$. For a matrix $M\in\mathcal{M}_{k,n}(\mathbb{F}_{q^m})$, the rank support of $M$, denoted by $\supp(M)$, is defined to be the linear space spanned by entries of $M$ over $\mathbb{F}_q$. Similarly, the rank weight of $M$ with respect to $\mathbb{F}_q$, denoted by $\tl{rk}(M)$, is defined as the dimension of $\supp(M)$ over $\mathbb{F}_q$.

\subsection{Gabidulin codes}
\noindent In this section, we will recall the concept of Gabidulin codes. Before doing this, we first introduce of the definition of Moore matrices and some related results.
\begin{definition}[Moore matrices]
For an integer $i$ and $\alpha\in\mathbb{F}_{q^m}$, we define $\alpha^{[i]}=\alpha^{q^i}$ to be the $i$-th Frobenius power of $\alpha$. For a vector $\bm{a}=(\alpha_1,\alpha_2,\ldots,\alpha_n)\in\mathbb{F}_{q^m}^n$, we define $\bm{a}^{[i]}=(\alpha_1^{[i]},\alpha_2^{[i]},\ldots,\alpha_n^{[i]})$ to be the $i$-th Frobenius power of $\bm{a}$. For positive integers $k\leqslant n$, a $k\times n$ Moore matrix induced by $\bm{a}$ is defined as
\[\tl{Mr}_k(\bm{a})=
\begin{pmatrix}
\alpha_1&\alpha_2&\cdots&\alpha_n\\
\alpha_1^{[1]}&\alpha_2^{[1]}&\cdots&\alpha_n^{[1]}\\
\vdots&\vdots&&\vdots\\
\alpha_1^{[k-1]}&\alpha_2^{[k-1]}&\cdots&\alpha_n^{[k-1]}
\end{pmatrix}.
\]
\end{definition}

For a positive integer $l$ and a matrix $M=(M_{ij})\in\mathcal{M}_{k,n}(\mathbb{F}_{q^m})$, we denote by $M^{[l]}=(M_{ij}^{[l]})$ the $l$-th Frobenius power of $M$. For a set $\mathcal{V}\subseteq\mathbb{F}_{q^m}^n$, we denote by $\mathcal{V}^{[l]}=\{\bm{v}^{[l]}:\bm{v}\in \mathcal{V}\}$ the $l$-th Frobenius power of $\mathcal{V}$. Particularly, for a linear code $\mathcal{C}\subseteq\mathbb{F}_{q^m}^n$, it is easy to verify that $\mathcal{C}^{[l]}$ is also a linear code over $\mathbb{F}_{q^m}$.

The following proposition presents simple properties about Moore matrices.
\begin{proposition}\label{proposition5}
\begin{itemize}
\item[\tl{(1)}]For two $k\times n$ Moore matrices $A,B\in\mathcal{M}_{k,n}(\mathbb{F}_{q^m})$, the sum $A+B$ is also a $k\times n$ Moore matrix.
\item[\tl{(2)}]For a Moore matrix $M\in\mathcal{M}_{k,n}(\mathbb{F}_{q^m})$ and a matrix $Q\in \mathcal{M}_{n,l}(\mathbb{F}_q)$, the product $MQ$ forms a $k\times l$ Moore matrix.
\item[\tl{(3)}]For a vector $\bm{a}\in\mathbb{F}_{q^m}^n$ with $\tl{rk}(\bm{a})=l$, there exist $\bm{a}'\in\mathbb{F}_{q^m}^l$ with $\tl{rk}(\bm{a}')=l$ and $Q\in \tl{GL}_n(\mathbb{F}_q)$ such that $\bm{a}=(\bm{a}',\bm{0})Q$. Furthermore, let $A=\tl{Mr}_k(\bm{a})$ and $A'=\tl{Mr}_k(\bm{a}')$, then $A=[A'|0]Q$.
\item[\tl{(4)}]For positive integers $k\leqslant n\leqslant m$, let $\bm{a}\in\mathbb{F}_{q^m}^n$ be a vector such that $\tl{rk}(\bm{a})=n$, then the Moore matrix $\tl{Mr}_k(\bm{a})$ has rank $k$.
\end{itemize}
\end{proposition}
\begin{proof}
Statements (1), (2) and (3) are trivial and therefore the proof is omitted here.
\begin{itemize}
\item[(4)]Let $\bm{a}=(\alpha_1,\cdots,\alpha_n)\in\mathbb{F}_{q^m}^n$. If $\rank(\tl{Mr}_k(\bm{a}))<k$, then there exists $\bm{\lambda}=(\lambda_0,\cdots,\lambda_{k-1})\in\mathbb{F}_{q^m}^k\backslash\{\bm{0}\}$ such that $\bm{\lambda}\tl{Mr}_k(\bm{a})=\bm{0}$. Let $f(x)=\sum_{j=0}^{k-1}\lambda_jx^{[j]}\in\mathbb{F}_{q^m}[x]$, then $f(\alpha_i)=0$ holds for any $1\leqslant i\leqslant n$. It follows that $f(\alpha)=0$ for any $\alpha\in\langle \alpha_1,\cdots,\alpha_n\rangle_q$, which conflicts with the fact that $f(x)=0$ admits at most $q^{k-1}$ roots.
\end{itemize}
\end{proof}

The following proposition states a fact that a Moore matrix can be decomposed as the product of a specific Moore matrix and a matrix over the base field. This fact was once exploited by Loidreau in \cite{loidreau2021analysis} to cryptanalyze an encryption scheme \cite{loidreau2017new} based on Gabidulin codes.
\begin{proposition}[Moore matrix decomposition]\label{proposition3}
Let $\bm{a}$ be a basis vector of $\mathbb{F}_{q^m}$ over $\mathbb{F}_q$. For positive integers $k\leqslant m$, let $M\in\mathcal{M}_{k,m}(\mathbb{F}_{q^m})$ be a Moore matrix generated by $\bm{a}$. Then for any $k\times n$ Moore matrix $M'\in\mathcal{M}_{k,n}(\mathbb{F}_{q^m})$, there exists $Q\in\mathcal{M}_{m,n}(\mathbb{F}_q)$ such that $M'=MQ$.
\end{proposition}

Now we formally introduce the definition of Gabidulin codes.
\begin{definition}[Gabidulin codes]\label{gabidulin}
For positive integers $k\leqslant n\leqslant m$, let $\bm{a}\in\mathbb{F}_{q^m}^n$ such that $\tl{rk}(\bm{a})=n$. The $[n,k]$ Gabidulin code generated by $\bm{a}$, denoted by $\tl{Gab}_{n,k}(\bm{a})$, is defined as the linear space spanned by rows of $\tl{Mr}_k(\bm{a})$ over $\mathbb{F}_{q^m}$. $\tl{Mr}_k(\bm{a})$ is called a standard generator matrix of $\tl{Gab}_{n,k}(\bm{a})$, and $\bm{a}$ a generating vector respectively.
\end{definition}

\begin{remark}
Gabidulin codes can be seen as a rank metric counterpart of generalized Reed-Solomon (GRS) codes, both of which admit good algebraic properties. The dual of an $[n,k]$ Gabidulin code is an $[n,n-k]$ Gabidulin code \cite{gaborit2018polynomial}. An $[n,k]$ Gabidulin code has minimum rank distance $n-k+1$ \cite{horlemann2015new} and can therefore correct up to $\left\lfloor\frac{n-k}{2}\right\rfloor$ rank errors in theory. Efficient decoding algorithms for Gabidulin codes can be found in \cite{gabidulin1985theory,loidreau2005welch,richter2004error}.
\end{remark}

To reduce the public key size, Lau and Tan exploited a so-called partial circulant matrix in the cryptosystem, as defined in the following.
\begin{definition}[Partial circulant matrices]
For a vector $\bm{a}=(\alpha_1,\alpha_2,\ldots,\alpha_n)\in\mathbb{F}_{q^m}^n$, the circulant matrix induced by $\bm{a}$, denoted by $\tl{Cir}_n(\bm{a})$, is defined to be a matrix whose first row is $\bm{a}$ and $i$-th row is obtained by cyclically right shifting the $i-1$-th row for $2\leqslant i\leqslant n$. The $k\times n$ partial circulant matrix induced by $\bm{a}$, denoted by $\tl{Cir}_k(\bm{a})$, is defined to be the first $k$ rows of $\tl{Cir}_n(\bm{a})$.
\end{definition}

\begin{remark}
For a normal basis vector $\bm{a}$ of $\mathbb{F}_{q^m}$ over $\mathbb{F}_q$, it is easy to verify that the $k\times n$ partial circulant matrix induced by $\bm{a}$ is exactly the $k\times n$ Moore matrix generated by $\bm{a}$. In other words, mathematically we have $\tl{Cir}_k(\bm{a})=\tl{Mr}_k(\bm{a})$.
\end{remark}

\section{Lau-Tan cryptosystem}\label{section3}
\noindent In this section, we mainly give a simple description of the Lau-Tan cryptosystem that uses Gabidulin codes as the underlying decodable code. For a given security level, choose positive integers $m>n>k>k'\geqslant 1$ and $r$ such that $k'=\lfloor\frac{k}{2}\rfloor$ and $r\leqslant \lfloor\frac{n-k}{2}\rfloor$. The Lau-Tan cryptosystem consists of the following three algorithms.
\begin{itemize}
\item Key Generation
\item[]Let $\mathcal{G}$ be an $[n,k]$ Gabidulin code over $\mathbb{F}_{q^m}$, and $G\in\mathcal{M}_{k,n}(\mathbb{F}_{q^m})$ be a generator matrix of $\mathcal{G}$ of standard form. Randomly choose matrices $S\in \tl{GL}_k(\mathbb{F}_{q^m})$ and  $T\in \tl{GL}_n(\mathbb{F}_q)$. Randomly choose $\bm{u}\in\mathbb{F}_{q^m}^n$ such that $\tl{rk}(\bm{u})=n$ and set $U=\tl{Cir}_k(\bm{u})$. Let $G_{pub}=SG+UT$, then we publish $(G_{pub},\bm{u})$ as the public key, and keep $(S,G,T)$ as the private key.
\item Encryption
\item[]For a plaintext $\bm{m}\in\mathbb{F}_{q^m}^{k'}$, randomly choose $\bm{m}_s\in\mathbb{F}_{q^m}^{k-k'}$ such that $\tl{rk}((\bm{m},\bm{m}_s)U)>\lceil\frac{3}{4}(n-k)\rceil$. Randomly choose $\bm{e}_1,\bm{e}_2\in\mathbb{F}_{q^m}^n$ such that $\tl{rk}(\bm{e}_1)=r_1\leqslant \frac{r}{2}$ and $\tl{rk}(\bm{e}_2)=r_2\leqslant \frac{r}{2}$. Compute $\bm{c}_1=(\bm{m},\bm{m}_s)U+\bm{e}_1$ and $\bm{c}_2=(\bm{m},\bm{m}_s)G_{pub}+\bm{e}_2$. Then the ciphertext is $\bm{c}=(\bm{c}_1,\bm{c}_2)$.
\item Decryption
\item[]For a ciphertext $\bm{c}=(\bm{c}_1,\bm{c}_2)\in\mathbb{F}_{q^m}^{2n}$, compute $\bm{c}'=\bm{c}_2-\bm{c}_1T=(\bm{m},\bm{m}_s)SG+\bm{e}_2-\bm{e}_1T$. Note that $\tl{rk}(\bm{e}_2-\bm{e}_1T)\leqslant \tl{rk}(\bm{e}_2)+\tl{rk}(\bm{e}_1T)\leqslant r$, decoding $\bm{c}'$ with the existing decoder of $\mathcal{G}$ will lead to $\bm{m}'=(\bm{m},\bm{m}_s)S$, then by computing $\bm{m}'S^{-1}$ one can recover the plaintext $\bm{m}$.
\end{itemize}

\section{Key recovery attack}\label{section4}
\noindent In this section, we will describe a method of efficiently recovering an equivalent private key of the Lau-Tan cryptosystem. We point out that the privacy of $T$ is of great importance for the security of the whole cryptosystem. Specifically, if one can find the secret $T$, then one can recover everything he needs to decrypt an arbitrary ciphertext in polynomial time. Before describing this attack, we first introduce some further results about Gabidulin codes.

\subsection{Further results about Gabidulin codes}
\noindent Similar to GRS codes in the Hamming metric, Gabidulin codes have good algebraic structure. For instance, if $\mathcal{G}$ is a Gabidulin code over $\mathbb{F}_{q^m}$, then its $l$-th Frobenius power is still a Gabidulin code. Formally, we introduce the following lemma.
\begin{lemma}\label{lemma1}
Let $\mathcal{G}$ be an $[n,k]$ Gabidulin code over $\mathbb{F}_{q^m}$, with $G\in\mathcal{M}_{k,n}(\mathbb{F}_{q^m})$ as a generator matrix. For any positive integer $l$, $\mathcal{G}^{[l]}$ is also an $[n,k]$ Gabidulin code and has $G^{[l]}$ as a generator matrix.
\end{lemma}
\begin{proof}
Trivial from a straightforward verification.
\end{proof}

For a proper positive integer $l$, the intersection of a Gabidulin code and its $l$-th Frobenius power is still a Gabidulin code, as described in the following proposition.
\begin{proposition}\label{proposition2}
For an $[n,k]$ Gabidulin code $\mathcal{G}$ over $\mathbb{F}_{q^m}$, let $\bm{g}\in\mathbb{F}_{q^m}^n$ be a generating vector of $\mathcal{G}$. For a positive integer $l\leqslant \min\{k-1,n-k\}$, the intersection of $\mathcal{G}$ and its $l$-th Frobenius power is an $[n,k-l]$ Gabidulin code with $\bm{g}^{[l]}$ as a generating vector. In other words, we have the following equality
\[\mathcal{G}\cap\mathcal{G}^{[l]}=\tl{Gab}_{n,k-l}(\bm{g}^{[l]}).\]
\end{proposition}
\begin{proof}
By the definition of Gabidulin codes, $\mathcal{G}$ is a linear space spanned by $\bm{g},\ldots,\bm{g}^{[k-1]}$ over $\mathbb{F}_{q^m}$, i.e. $\mathcal{G}=\langle\bm{g},\ldots,\bm{g}^{[k-1]}\rangle_{q^m}$. By Lemma \ref{lemma1}, we have $\mathcal{G}^{[l]}=\langle\bm{g}^{[l]},\ldots,\bm{g}^{[k+l-1]}\rangle_{q^m}$. Note that $l\leqslant \min\{k-1,n-k\}$, then $k+l\leqslant n$ and $\bm{g},\ldots,\bm{g}^{[k+l-1]}$ are linearly independent over $\mathbb{F}_{q^m}$. It follows that $\mathcal{G}\cap\mathcal{G}^{[l]}=\langle\bm{g}^{[l]},\ldots,\bm{g}^{[k-l-1]}\rangle_{q^m}$ forms an $[n,k-l]$ Gabidulin code, having $\bm{g}^{[l]}$ as a generating vector. This completes the proof.
\end{proof}

\begin{lemma}\label{lemma2}
For positive integers $k<n$, let $\mathcal{G}\subset\mathbb{F}_{q^m}^n$ be an $[n,k]$ Gabidulin code, and $A\in\mathcal{M}_{k,n}(\mathbb{F}_{q^m})$ be a nonzero Moore matrix. If all the row vectors of $A$ are codewords in $\mathcal{G}$, then $A$ must be a generator matrix of $\mathcal{G}$.
\end{lemma}
\begin{proof}
It suffices to prove $\rank(A)=k$. Suppose that $A$ is generated by $\bm{a}\in\mathbb{F}_{q^m}^n$, i.e. $A=\tl{Mr}_k(\bm{a})$. Let $l=\tl{rk}(\bm{a})$, then there exist $\bm{a}'\in\mathbb{F}_{q^m}^l$ with $\tl{rk}(\bm{a}')=l$ and $Q\in \tl{GL}_n(\mathbb{F}_q)$ such that $\bm{a}=(\bm{a}',\bm{0})Q$. Let $A'\in\mathcal{M}_{k,l}(\mathbb{F}_{q^m})$ be a Moore matrix generated by $\bm{a}'$, then it follows immediately that $A=[A'|0]Q$. If $l>k$, then $\rank(A)=\rank(A')=k$ due to Proposition \ref{proposition5} and therefore the conclusion is proved. Otherwise, there will be $\langle A'\rangle_{q^m}=\mathbb{F}_{q^m}^l$. From this we can deduce that the minimum rank distance of $\mathcal{G}$ will be $1$, which conflicts with the fact that $\tl{rk}(\mathcal{G})=n-k+1\geqslant 2$. Hence $l>k$ and $\rank(A)=k$. This completes the proof.
\end{proof}

By Definition \ref{gabidulin}, a Gabidulin code is uniquely determined by its generating vector. Naturally, it is important to make clear what all these vectors look like and how many generating vectors there exist for a Gabidulin code.
\begin{proposition}\label{proposition1}
Let $\mathcal{G}$ be an $[n,k]$ Gabidulin code over $\mathbb{F}_{q^m}$, with $\bm{g}\in\mathbb{F}_{q^m}^n$ as a generating vector. Let $\bm{g}'\in\mathbb{F}_{q^m}^n$ be a codeword in $\mathcal{G}$, then $\bm{g}'$ forms a generating vector if and only if there exists $\gamma\in\mathbb{F}_{q^m}^*$ such that $\bm{g}'=\gamma\bm{g}$.
\end{proposition}
\begin{proof}
Assume that $\bm{g}=(\alpha_1,\ldots,\alpha_n)$ and $\bm{g}'=(\alpha'_1,\ldots,\alpha'_n)$, let $G=\tl{Mr}_k(\bm{g})$ and $G'=\tl{Mr}_k(\bm{g}')$. The conclusion is trivial if $\bm{g}=\bm{g}'$. Otherwise, without loss of generality we assume that $\alpha'_1\neq\alpha_1$, then there exists $\gamma\in\mathbb{F}_{q^m}^*\backslash\{1\}$ such that $\alpha'_1=\gamma\alpha_1$. Let 
\[S=
\begin{pmatrix}
\gamma&0&\cdots&0\\
0&\gamma^{[1]}&\cdots&0\\
\vdots&\vdots&&\vdots\\
0&0&\cdots&\gamma^{[k-1]}
\end{pmatrix},
\]
then $SG=\tl{Mr}_k(\gamma\bm{g})$. Let $\bm{g}^*=\gamma\bm{g}-\bm{g}'=(0,\gamma\alpha_2-\alpha'_2,\ldots,\gamma\alpha_n-\alpha'_n)$ and $G^*=\tl{Mr}_k(\bm{g}^*)$, then $G^*=SG-G'$. Apparently all the row vectors of $G^*$ are codewords in $\mathcal{G}$. If $\bm{g}^*\neq \bm{0}$, then $G^*$ forms a generator matrix of $\mathcal{G}$ of standard form due to Lemma \ref{lemma2}. Together with $\tl{rk}(\bm{g}^*)\leqslant n-1$, easily we can deduce that $\tl{rk}(\bm{c})\leqslant n-1$ for any $\bm{c}\in\mathcal{G}$, which clearly contradicts the fact that $\tl{rk}(\bm{g})=n$. Therefore there must be $\bm{g}^*=\bm{0}$, or equivalently $\bm{g}'=\gamma\bm{g}$. The opposite is obvious from a straightforward verification.
\end{proof}

The following corollary is drawn immediately from Proposition \ref{proposition1}.
\begin{corollary}\label{corollary1}
An $[n,k]$ Gabidulin code over $\mathbb{F}_{q^m}$ admits $q^m-1$ generator matrices of standard form, or equivalently $q^m-1$ generating vectors.
\end{corollary}

\subsection{Recovering the secret $T$}
\noindent In this section, we mainly describe an efficient algorithm for recovering the secret $T$. In summary, the technique we adopt here is to convert the problem of recovering $T$ into solving a multivariate linear system, which clearly costs polynomial time. Before doing this, we first introduce the so-called subfield expanding transform.

\textbf{Subfield Expanding Transform.} For $\beta_1,\ldots,\beta_n\in\mathbb{F}_{q^m}$, we construct an equation as 
\begin{align}\label{equation5}
\sum_{j=1}^nx_j\beta_j=0,
\end{align} 
where $x_j$'s are underdetermined variables in $\mathbb{F}_q$. Let $\bm{a}$ be a basis vector of $\mathbb{F}_{q^m}$ over $\mathbb{F}_q$. For each $1\leqslant j\leqslant n$, there exists $\bm{b}_j\in\mathbb{F}_q^m$ such that $\beta_j=\bm{b}_j\bm{a}^T$. It follows that $\sum_{j=1}^nx_j\beta_j=\sum_{j=1}^nx_j(\bm{b}_j\bm{a}^T)=(\sum_{j=1}^nx_j\bm{b}_j)\bm{a}^T$, and moreover, (\ref{equation5}) holds if and only if 
\begin{align}\label{equation6}
\sum_{j=1}^nx_j\bm{b}_j=\bm{0}.
\end{align}
Obviously, the linear systems (\ref{equation5}) and (\ref{equation6}) share the same solution space. A transform that derives (\ref{equation6}) from (\ref{equation5}) is called a subfield expanding transform.

In the Lau-Tan cryptosystem, let $H\in\mathcal{M}_{n-k,n}(\mathbb{F}_{q^m})$ be a parity-check matrix of $\mathcal{G}$ of standard form. Let $M\in\mathcal{M}_{n-k,m}(\mathbb{F}_{q^m})$ be a Moore matrix generated by a basis vector of $\mathbb{F}_{q^m}$ over $\mathbb{F}_q$, then there exists an underdetermined matrix $X\in\mathcal{M}_{m,n}(\mathbb{F}_q)$ such that $H=MX$. Let $T^*\in \tl{GL}_n(\mathbb{F}_q)$ be another underdetermined matrix such that $G_{pub}-\tl{Cir}_k(\bm{u})T^*=G_{pub}-UT^*$ forms a generator matrix of $\mathcal{G}$, or equivalently
\begin{align}\label{equation1}
(G_{pub}-UT^*)(MX)^T=G_{pub}X^TM^T-UT^*X^TM^T=0.
\end{align}
We therefore obtain a system of $k(n-k)$ multivariate quadratic equations, with $n(m+n)$ variables in $\mathbb{F}_q$. This  system admits at least $q^m$ solutions. Specifically, we introduce the following proposition.

\begin{proposition}
The linear system \tl{(\ref{equation1})} has at least $q^m$ solutions. 
\end{proposition}
\begin{proof}
If $T^*=T$, then we can deduce from (\ref{equation1}) that
\[
(G_{pub}-UT^*)(MX)^T=G_{pub}X^TM^T-UT^*X^TM^T=SGX^TM^T=SG(MX)^T=0.
\]
Note that $SG\in\mathcal{M}_{k,n}(\mathbb{F}_{q^m})$ forms a generator matrix of $\mathcal{G}$. By $SG(MX)^T=0$, we conclude that all the row vectors of $MX$ are contained in the dual code of $\mathcal{G}$, which is an $[n,n-k]$ Gabidulin code. On the other hand, it is easy to see that $MX$ is an $(n-k)\times n$ Moore matrix. By Lemma \ref{lemma2}, $MX$ forms a standard generator matrix of $\mathcal{G}^\perp$ for a nonzero $X$. Then the conclusion is immediately proved from Corollary \ref{corollary1}. Furthermore, we have that $X$ is an $m\times n$ matrix of full rank.
\end{proof}

Note that solving a multivariate quadratic system generally requires exponential time. Instead of solving the system (\ref{equation1}) directly, the technique we exploit here is to consider each entry of $T^*X^T$ as a new variable in $\mathbb{F}_{q}$ and set $Y=X{T^*}^T$. In other words, we rewrite (\ref{equation1}) into the following matrix equation
\begin{align}\label{equation2}
G_{pub}X^TM^T-UY^TM^T=0.
\end{align}
This enables us to obtain a linear system of $k(n-k)$ equations, with coefficients in $\mathbb{F}_{q^m}$ and $2mn$ variables in $\mathbb{F}_q$. To solve the system (\ref{equation2}), we usually convert this problem into an instance over the base field $\mathbb{F}_q$. Applying the subfield expanding transform to (\ref{equation2}) leads to a linear system of $mk(n-k)$ equations over $\mathbb{F}_q$, with $2mn$ variables to be determined. For a cryptographic use, generally we have $mk(n-k)\geqslant 2mn$.

\begin{remark}\label{remark1}
For each solution $(X,T^*)$ of (\ref{equation1}), one can easily obtain a solution of (\ref{equation2}) by computing $Y=X{T^*}^T$, which implies that there are also at least $q^m$ solutions for (\ref{equation2}). Conversely, if (\ref{equation2}) has exactly $q^m$ solutions, then these solutions must correspond to solutions of (\ref{equation1}) where $T^*=T$. In this situation,  for any nonzero solution $(X,Y)$ of (\ref{equation2}), solving the matrix equation $Y=X{T^*}^T$ will lead to the secret $T=T^*$. 
\end{remark}

As for whether or not the system (\ref{equation2}) has other types of solutions, we make an \textbf{Assumption} that the answer is negative. According to our experimental results on Magma, this assumption holds with high probability. To make it easier,  a simplified version of this problem is considered. Let $G$ be an arbitrary generator matrix of an $[n,k]$ Gabidulin code and $\bm{u}\in\mathbb{F}_{q^m}^n$ such that $\tl{rk}(\bm{u})=n$. We then construct a matrix equation as 
\[
GX^TM^T+\tl{Cir}_k(\bm{u})Y^TM^T=0,
\] 
where $M\in\mathcal{M}_{n-k,m}(\mathbb{F}_{q^m})$ is a Moore matrix generated by a basis vector of $\mathbb{F}_{q^m}$ over $\mathbb{F}_q$ and $X,Y\in\mathcal{M}_{m,n}(\mathbb{F}_q)$ are two underdetermined matrices. By applying the subfield expanding transform to this system above, we obtain a new system over $\mathbb{F}_q$. By Remark \ref{remark1}, if this newly obtained system admits a solution space of dimension $m$, then there must be $Y=0$. Eventually we ran 1000 random tests for $q=2, m=25, n=23$ and $k=10$. It turns out that the assumption holds in all of these random instances.

\begin{algorithm}[H]
  \small
  \caption{: $T$-recovering algorithm}
  \label{algorithm1}
  \hspace*{0.02in} {\bf Input:}
    The public key $(G_{pub},\bm{u})$\\
  \hspace*{0.02in} {\bf Output:}
    The secret $T$
  \begin{algorithmic}[1]
  \State Let $\bm{a}$ be a basis vector of $\mathbb{F}_{q^m}$ over $\mathbb{F}_q$ and set $M=\tl{Mr}_{n-k}(\bm{a})$
  \State Let $X,Y\in\mathcal{M}_{m,n}(\mathbb{F}_q)$ be two underdetermined matrices and construct a linear system 
    \begin{align}\label{equation3}
      G_{pub}X^TM^T-\tl{Cir}_k(\bm{u})Y^TM^T=0
    \end{align}
  \State Applying the subfield expanding transform to (\ref{equation3}) to obtain a linear system over $\mathbb{F}_q$
  \State Solve this system for $(X,Y)$
  \State For any nonzero $(X,Y)$, solve the matrix equation $Y=X{T^*}^T$ for $T^*$
  \State \Return $T=T^*$
  \end{algorithmic}
\end{algorithm}

\subsection{Finding an equivalent $(S',G')$}
\noindent With the knowledge of a generating vector, we can deduce many characteristics of a Gabidulin code, such as an efficient decoding algorithm. A natural question is how to derive the generating vector of a Gabidulin code from an arbitrary generator matrix. In \cite{horlemann2018extension} the authors presented an iterative method of computing the generating vector. Here in this paper we present a different approach to do this, as described in the following.

\textbf{An approach to compute the generating vector}. For an $[n,k]$ Gabidulin code $\mathcal{G}$ over $\mathbb{F}_{q^m}$, let $G\in\mathcal{M}_{k,n}(\mathbb{F}_{q^m})$ be an arbitrary generator matrix of $\mathcal{G}$. We first compute a parity-check matrix of $\mathcal{G}$ from $G$, say $H$. Let $M\in\mathcal{M}_{k,m}(\mathbb{F}_{q^m})$ be a Moore matrix generated by a basis vector of $\mathbb{F}_{q^m}$ over $\mathbb{F}_q$, then there exists an underdetermined matrix $X\in\mathcal{M}_{m,n}(\mathbb{F}_q)$ such that $MX$ forms a  standard generator matrix of $\mathcal{G}$. By setting $(MX)H^T=0$ we obtain a linear system of $k(n-k)$ equations, with coefficients in $\mathbb{F}_{q^m}$ and $mn$ variables in $\mathbb{F}_q$. Applying the subfield expanding transform to this system leads to a new linear system over the base field $\mathbb{F}_q$, with $mk(n-k)$ equations and $mn$ variables. For a cryptographic use, generally we have $mk(n-k)\geqslant mn$. By Corollary \ref{corollary1}, this newly obtained system admits $q^m-1$ nonzero solutions. And for any nonzero solution, say $X$, the first row of $MX$ will be a generating vector of $\mathcal{G}$.
\begin{algorithm}[H]
  \small
  \caption{: Finding an equivalent $(S',G')$}
  \label{algorithm2}
  \hspace*{0.02in} {\bf Input:}
    $(G_{pub},\bm{u},T)$\\
  \hspace*{0.02in} {\bf Output:}
    $(S',G')$ such that $G'$ forms a standard generator matrix of $\mathcal{G}$ and $S'G'=SG$
  \begin{algorithmic}[1]
  \State Let $\bm{a}$ be a basis vector of $\mathbb{F}_{q^m}$ over $\mathbb{F}_q$ and construct a Moore matrix $M=\tl{Mr}_k(\bm{a})$
  \State Compute $SG=G_{pub}-\tl{Cir}_k(\bm{u})T$ and $\mathcal{G}=\langle SG\rangle_{q^m}$
  \State Let $H\in\mathcal{M}_{n-k,n}(\mathbb{F}_{q^m})$ be a parity-check matrix of $\mathcal{G}$
  \State Let $X\in\mathcal{M}_{m,n}(\mathbb{F}_q)$ be an underdetermined matrix and construct a linear system as 
    \begin{align}\label{equation4}
      (MX)H^T=0
    \end{align}
  \State Applying the subfield expanding transform to (\ref{equation4}) to obtain a new system over $\mathbb{F}_q$
  \State Solve this new system for a nonzero $X$ and compute $G'=MX$
  \State Compute $S'\in \tl{GL}_k(\mathbb{F}_{q^m})$ such that $S'G'=SG$
  \State \Return $(S',G')$
  \end{algorithmic}
\end{algorithm}

\subsection{Complexity of the attack}
\noindent Our attack consists of two phases: firstly, we manage to recover the secret $T$ from the published information, as described in Algorithm \ref{algorithm1}; secondly, with the knowledge of $T$ and the public key, we compute a standard generator matrix $G'$ of the secret Gabidulin code and an invertible matrix $S'$, as described in Algorithm \ref{algorithm2}. Hence the complexity analysis is done in the following two aspects.

\textbf{Complexity of Algorithm} \ref{algorithm1}. In Step 1 we construct a Moore matrix $M\in\mathcal{M}_{n-k,m}(\mathbb{F}_{q^m})$ whose first row vector forms a basis of $\mathbb{F}_{q^m}$ over $\mathbb{F}_q$. To avoid executing the Frobenius operation, here we choose $\bm{a}$ to be a normal basis vector, then we set $M=\tl{Cir}_{n-k}(\bm{a})$. In Step 2 we construct a multivariate linear system by performing matrix multiplication, requiring $\mathcal{O}(mn^3)$ operations in $\mathbb{F}_{q^m}$. The subfield expanding transform performed to (\ref{equation3}) requires $\mathcal{O}(m^3n^3)$ operations in $\mathbb{F}_{q^m}$. Step 4 requires $\mathcal{O}(m^3n^3)$ operations to solve the linear system over $\mathbb{F}_q$ and Step 5 requires $\mathcal{O}(n^3)$ operations in $\mathbb{F}_q$. The total complexity of Algorithm \ref{algorithm1} consists of $\mathcal{O}(m^3n^3+mn^3)$ operations in $\mathbb{F}_{q^m}$ and $\mathcal{O}(m^3n^3+n^3)$ operations in $\mathbb{F}_q$.

\textbf{Complexity of Algorithm} \ref{algorithm2}. In Step 1 we still choose a normal basis vector to construct $M$. To compute $SG$, we perform matrix addition and multiplication with $\mathcal{O}(n^3)$ operations in $\mathbb{F}_{q^m}$. Step 3 computes a parity-check $H$ of $\mathcal{G}$ from $SG$, requiring $\mathcal{O}(n^3)$ operations in $\mathbb{F}_{q^m}$. Then we construct a linear system in Step 4, which costs $\mathcal{O}(mn^3)$ operations in $\mathbb{F}_{q^m}$. In Step 5 we apply the subfield expanding transform to (\ref{equation4}) to obtain a new system over $\mathbb{F}_q$, requiring $\mathcal{O}(m^3n^3)$ operations in $\mathbb{F}_{q^m}$. Solving this new system in Step 6 costs $\mathcal{O}(m^3n^3)$ operations in $\mathbb{F}_q$, and compute $G'=MX$ with $\mathcal{O}(mn^2)$ operations in $\mathbb{F}_{q^m}$. In Step 7, we shall compute $S'$ from $S'G'$ with $\mathcal{O}(n^3)$ operations. The total complexity of Algorithm \ref{algorithm2} consists of $\mathcal{O}(m^3n^3+mn^3+n^3)$ operations in $\mathbb{F}_{q^m}$ and $\mathcal{O}(m^3n^3)$ operations in $\mathbb{F}_q$.

Finally, the total complexity of the attack is $\mathcal{O}(m^3n^3+mn^3+n^3)$ in $\mathbb{F}_{q^m}$ plus $\mathcal{O}(m^3n^3+n^3)$ in $\mathbb{F}_q$.

\subsection{Implementation}
\noindent This attack has been implemented on Magma and permits to recover the secret $T$. We tested this attack on a personal computer and succeeded for parameters as illustrated in Table \ref{table1}. For each parameter set, the attack has been run 100 times and the last column gives the average timing (in seconds). Our implementation is just a proof of the feasibility of this attack and does not consider the proposed parameters in \cite{lau2018new} due to limited resources. 

\begin{table}[H]  
\parbox{.45\textwidth}{\caption{These experiments were performed using Magma V2.11-1 on an 11th Gen Intel(R) Core(TM) i7-11700 @ 2.5GHz processor with 16 GB of memory.}\label{table1}}  
\hspace{\fill}
\parbox{.55\textwidth}{
\begin{center}
 \setlength{\tabcolsep}{6mm}{
\begin{tabular}{lllr|r}
\hline\noalign{\smallskip}
$q$ & $m$ & $n$ & $\makecell*[c]{k}$ & $\makecell*[c]{t}$ \\ 
\noalign{\smallskip}\hline\noalign{\smallskip}
2 & 22 & 18 & 9 & 8.6 \\  
2 & 28 & 22 & 9 & 40.7 \\ 
2 & 35 & 26 & 12 & 173.2 \\
\noalign{\smallskip}\hline
\end{tabular}}
\end{center}}
\end{table}

\section{Conclusion}\label{section5}
\noindent Our attack revealed the structural weakness of the Lau-Tan cryptosystem. Although the first part of the public key hid the structure of Gabidulin codes perfectly, the second part did reveal important information that can be used to design a key recovery attack. Specifically, we convert the problem of recovering the private key into solving a multivariate linear system over the base field $\mathbb{F}_q$. Even though this system admits a solution space of dimension $m$, we are able to recover the secret $T$ and then an equivalent $(S',G')$ from any nonzero solution. Extensive experiments have been performed and the results show that our attack accords with the theoretical expectations. In summary, we found a polynomial-time key recovery attack on the Lau-Tan cryptosystem under a reasonable assumption. 

\begin{acknowledgements}
This research is supported by the National Key Research and Development Program of China (Grant No. 2018YFA0704703), the National Natural Science Foundation of China (Grant No. 61971243), the Natural Science Foundation of Tianjin (20JCZDJC00610), and the Fundamental Research Funds for the Central Universities of China (Nankai University).
\end{acknowledgements}

\end{document}